\documentclass[11pt]{article}
\pdfoutput=1

\usepackage{fullpage}
\usepackage{graphicx}   
\usepackage{subfig}  
\usepackage{hyperref}   
\usepackage{mdwlist}
\usepackage{xspace}
\usepackage[usenames,dvipsnames]{color}
\usepackage{amsmath, amsthm,amssymb}   
\usepackage{tabu}

\newcommand{\cut}[1]{}

\newcommand{\introparagraph}[1]{\textbf{#1.}}        

\allowdisplaybreaks 

\usepackage{aliascnt}  		

\newtheorem{theorem}{Theorem}[section]          	
\newaliascnt{lemma}{theorem}				
\newtheorem{lemma}[lemma]{Lemma}              	
\aliascntresetthe{lemma}  					
\newaliascnt{conjecture}{theorem}			
\aliascntresetthe{conjecture}  				
\newaliascnt{remark}{theorem}				
\aliascntresetthe{remark}  					
\newaliascnt{corollary}{theorem}			
\newtheorem{corollary}[corollary]{Corollary}      
\aliascntresetthe{corollary}  				
\newaliascnt{definition}{theorem}			
\aliascntresetthe{definition}  				
\newaliascnt{proposition}{theorem}			
\aliascntresetthe{proposition}  				
\newaliascnt{example}{theorem}			
\newtheorem{example}[example]{Example}  	
\aliascntresetthe{example}  				

\providecommand{\bx}[0]{\mathbf{x}}
\providecommand{\by}[0]{\mathbf{y}}
\providecommand{\bM}[0]{\mathbf{M}}

\providecommand{\bv}[0]{\mathbf{v}}
\providecommand{\bu}[0]{\mathbf{u}}

\providecommand{\ba}[0]{\mathbf{a}}

\providecommand{\E}[0]{\mathbb{E}}

\newcommand{\set}[1]{\{#1\}}                    
\newcommand{\setof}[2]{\{{#1}\mid{#2}\}}        

\newcommand{\mA}[0]{\mathcal{A}} 
\newcommand{\mB}[0]{\mathcal{B}}

\newcommand{\mM}[0]{\mathcal{M}}

\newcommand{\msg}[0]{{\textsf{msg}}}

\newcommand{\mpc}[0]{MPC}

\newcommand{\vars}[1]{\textsf{vars}(#1)}
\newcommand{\atoms}[1]{\textsf{atoms}(#1)}

\newcommand{\pk}[1]{\textsf{pk}(#1)}
\newcommand{\qpk}[1]{\textsf{pk}^+(#1)}

\usepackage{mathpazo} 
\usepackage[scaled]{helvet} 
\usepackage{courier} 
\normalfont
\usepackage[T1]{fontenc}
\usepackage{authblk}

\hypersetup{colorlinks,
 citecolor=blue,
  linkcolor=black,
  urlcolor=black}

\begin{document}

\title{Worst-Case Optimal Algorithms for Parallel Query Processing\thanks{This work is partially supported by NSF IIS-1247469,  AitF 1535565, CCF-1217099 and CCF-1524246.}}

\author[*]{Paul Beame}
\author[$\dag$]{Paraschos Koutris}
\author[*]{Dan Suciu}
\affil[*]{University of Washington, Seattle, WA}
\affil[$\dag$]{University of Wisconsin-Madison, Madison, WI}

\maketitle

\begin{abstract}
In this paper, we study the communication complexity for the problem of
computing a conjunctive query on a large database in a parallel setting with
$p$ servers. In contrast to previous work, where upper and lower bounds on
the communication were specified for particular structures of data (either
data without skew, or data with specific types of skew), in this work we
focus on {\em worst-case analysis} of the communication cost. The goal is to find
worst-case optimal parallel algorithms, similar to the work of~\cite{NPRR12}
for sequential algorithms. 

We first show that for a single round we can obtain an optimal worst-case 
algorithm. The optimal load for a conjunctive query $q$ when all relations have
size equal to $M$ is $O(M/p^{1/\psi^*})$, where $\psi^*$ is a new query-related quantity
called the {\em edge quasi-packing number}, which is different from both the
edge packing number and edge cover number of the query hypergraph.
For multiple rounds, we present algorithms that are optimal for 
several classes of queries. Finally, we show a surprising 
connection to the external memory model, which allows us to translate
parallel algorithms to external memory algorithms. This technique allows
us to recover (within a polylogarithmic factor) several recent results on the
I/O complexity for computing join queries, and also obtain optimal
algorithms for other classes of queries.  
\end{abstract}

\section{Introduction}
\label{sec:intro}

The last decade has seen the development and widespread use of massively parallel systems that perform 
data analytics tasks over big data: examples of such systems are MapReduce~\cite{DBLP:conf/osdi/DeanG04},
Dremel~\cite{dremel}, Spark~\cite{spark:2012} and Myria~\cite{myria}. In contrast to traditional database
systems, where the computational complexity is dominated by the disk access time, the data now typically fits in main 
memory, and the dominant cost becomes that of communicating data and synchronizing among the servers in the cluster.

In this paper, we present a {\em worst-case analysis} of algorithms
for processing of conjunctive queries (multiway join queries) on such massively
parallel systems. Our analysis is based on the Massively Parallel Computation model, 
or \mpc~\cite{BKS13,BKS14}. \mpc\ is a theoretical model where the computational
complexity of an algorithm is characterized by both the {\em number of rounds} (so the number
of synchronization barriers) and 
the maximum amount of data, or {\em maximum load}, that each processor receives at
every round.

The focus of our analysis on worst-case behavior of algorithms is a fundamentally
different approach from previous work, where optimality of a parallel algorithm
was defined for a specific input, or a specific family of inputs. Here we obtain upper 
bounds on the load of the algorithm across all possible types of input data. 
To give a concrete example, consider the simple join between two binary relations $R$ and $S$ 
of size $M$ in bits (and $m$ tuples), denoted $q(x,y,z) = R(x,z), S(y,z)$, and suppose that
the number of servers is $p$. 
In the case where there is no data skew (which means in our case that the frequency of 
each value of the $z$ variable in both $R$ and $S$ is at most $m/p$), it has been shown in
\cite{BKS14} that the join can be computed in a single round with load 
$\tilde{O}(M/p)$ (where the notation $\tilde{O}$ hides a polylogarithmic factor depending on $p$),
by simply hashing each tuple according to the value
of the $z$ variable. However, if the $z$ variable is heavily skewed both
in $R$ and $S$ (and in particular if there exists a single value of $z$),
computing the query becomes equivalent to computing a cartesian product, 
for which we need $\Omega(M/p^{1/2})$ load. In this scenario, although for
certain instances we can obtain better guarantees for the load, the heavily skewed instance is a
{\em worst-case input}, in the sense that the lower bound $\Omega(M/p^{1/2})$ specifies
the worst possible load that we may encounter. Our goal is to design algorithms for single or
multiple rounds that are optimal with respect to such worst-case inputs and
never incur larger load for any input. \\[-2mm]

\noindent \textbf{Related Work.}
Algorithms for joins in the \mpc\ model were previously analyzed
in~\cite{BKS13, BKS14}. In~\cite{BKS13}, the authors presented algorithms for one and multiple
rounds on input data without skew (in particular when each value appears exactly once in
each relation, which is a called a {\em matching database}). 
In~\cite{BKS14}, the authors showed that the HyperCube (HC) algorithm, first
presented by Afrati and Ullman~\cite{DBLP:conf/edbt/AfratiU10}, can optimally compute any
conjunctive query for a single round on data without skew. The work in~\cite{BKS14} also presents one-round
algorithms and lower bounds for skewed data but the upper and lower bounds do not necessarily coincide.

Several other computation models have been proposed in order to understand
the power of MapReduce and related massively parallel programming
paradigms~\cite{DBLP:journals/talg/FeldmanMSSS10,DBLP:conf/soda/KarloffSV10,DBLP:conf/pods/KoutrisS11,DBLP:journals/corr/abs-1206-4377}. All these models identify the number of communication steps/rounds as a
main complexity parameter, but differ in their treatment of the
communication. Previous work~\cite{WZ13,Klauck15} has also focused on computing 
various graph problems in message-passing parallel models. In contrast to this work, where
we focus on algorithms that require a constant number of rounds, the authors consider algorithms
that need a large number of rounds.

Our setting and worst-case analysis can be viewed as the analogous version of the
work of Ngo et al.~\cite{NPRR12} on worst-case optimal algorithms for multiway join 
processing. As we will show later, the worst-case instances for a given query $q$ are
different for the two settings in the case of one round, but coincide for all the families of
queries we examine when we consider multiple rounds. \\[-2mm]

\noindent \textbf{Our Contributions.}
We first present in Section~\ref{sec:one-step} tight upper and lower bounds for the 
worst-case load of one-round algorithms for any full conjunctive query $q$ without 
self-joins.\footnote{The restriction to queries without self-joins is not limiting, since we
can extend our result to queries with self-joins (by losing a constant factor)
by treating copies of a relation as distinct relations.
The parallel complexity for queries with projections is however an open question.} 
The optimal algorithm
uses a different parametrization (share allocation) of the HyperCube algorithm for different parts
of the input data, according to the values that are skewed. In the case where all relation sizes
are equal to $M$, the algorithm
achieves an optimal load $\tilde{O}(M/p^{1/\psi^*(q)})$, where $\psi^*(q)$ is the {\em edge quasi-packing number}
of the query $q$. An edge quasi-packing is an edge packing on any vertex-induced projection of
the query hypergraph (in which we shrink hyperedges when we remove
vertices). 

In Section~\ref{sec:multi-step}, we show that for any full conjunctive query $q$, any algorithm with a constant number of rounds
requires a load of $\Omega(M/p^{1/\rho^*})$, where $\rho^*$ is the {\em edge cover number}.
We then present optimal (within a polylogarithmic factor) multi-round algorithms for several classes of join queries.
Our analysis shows that some queries (such as the star query
$T_k$) can be optimally computed using the optimal single-round algorithm from
Section~\ref{sec:one-step}. However, other classes of queries, such as the cycle query
$C_k$ for $k\ne 4$, the line query $L_k$, or the Loomis-Whitney join $LW_k$ require 2 or more rounds to achieve the optimal load. For example, we present an algorithm for the full query (or clique) $K_k$ that uses $k-1$ rounds to achieve the optimal load (although it is open whether only 2 rounds are sufficient).

Finally, in Section~\ref{sec:ex-memory} we present a surprising application of our results in 
the setting of external memory algorithms. In this setting, the input data does not fit into main memory,
and the dominant cost of an algorithm is the I/O complexity: reading the data from the disk into 
the memory and  writing data on the disk.
 In particular, we show that we can simulate an \mpc\ algorithm in the 
external memory setting, and obtain almost-optimal (within a polylogarithmic factor)
external memory algorithms for computing triangle queries; the same technique
can be easily applied to other classes of queries.  

\section{Background}
\label{sec:background}

In this section, we introduce the \mpc\ model and present the necessary terminology and
technical tools that we will use later in the paper.

\subsection{The MPC Model}

We first review the {\em Massively Parallel Computation model (MPC)}, which allows us
to analyze the performance of algorithms in parallel environments.  In the
\mpc\ model, computation is performed by $p$ servers, or processors, 
connected by a complete network of private channels. 
The computation proceeds in {\em steps}, or {\em rounds}, where each round consists of 
two distinct phases. In the {\em communication phase}, the 
servers exchange data, each by communicating with all other servers.
In the {\em computation phase}, each server performs only local computation.

The input data of size $M$ bits is initially uniformly partitioned among the $p$ servers, that is, 
each server stores $M/p$ bits of data. 
At the end of the execution, the output must be present in the union of the output of the $p$ processors. 

The execution of a parallel algorithm in the MPC model is captured by two
parameters. The first parameter is the {\em number of rounds} $r$ that
the algorithm requires. The second parameter is the {\em maximum load} $L$,
which measures the maximum amount of data (in bits) received by any server during any
round.

All the input data will be distributed during some round, since we need to perform some computation
on it. Thus, at least one server will receive at least data of size $M/p$.
On the other hand, the maximum load will never exceed $M$, since
any problem can be trivially solved in one round by simply sending the
entire data to server 1, which can then compute the answer locally.  Our
typical loads will be of the form $M/p^{1-\varepsilon}$, for some parameter $\varepsilon$ 
($0\leq \varepsilon < 1$) that depends on the query.  For a similar
reason, we do not allow the number of rounds to reach $r = p$, because
any problem can be trivially solved in $p$ rounds by sending 
$M/p$ bits of data at each round to server 1, until this server accumulates
the entire data.  In this paper we only consider the case $r = O(1)$.

\subsection{Conjunctive Queries}

In this paper we focus on a particular class of problems for the \mpc\
model, namely computing answers to conjunctive queries over a
database. We fix an input vocabulary $S_1, \ldots, S_\ell$, where
each relation $S_j$ has a fixed arity $a_j$; we denote $a = \sum_{j
  =1}^{\ell} a_j$. The input data consists of one relation instance
for each symbol.  

We consider full conjunctive queries (CQs) without self-joins, denoted
as follows:
$$q(x_1, \dots, x_k) = S_1(\dots), \dots, S_\ell(\dots).$$
The query is {\em full}, meaning that every variable in the body
appears in the head (for example $q(x) = S(x,y)$ is not full), and
{\em without self-joins}, meaning that each relation name $S_j$
appears only once (for example $q(x,y,z) = S(x,y), S(y,z)$ has a
self-join). 
We use $\vars{S_j}$ to denote the set of
variables in the atom $S_j$, and $\vars{q}$ to denote the set of variables
in all atoms of $q$. 
Further, 
$k$ and $\ell$ denote the number of variables and atoms in $q$ respectively.
The {\em hypergraph} of a conjunctive query $q$ is defined by introducing one node
for each variable in the body and one hyperedge for each set of
variables that occur in a single atom.

The {\em fractional edge packing} associates a non-negative weight $u_j$ to each
atom $S_j$ such that for every variable $x_i$, the sum of the weights for
the atoms that contain $x_i$ does not exceed 1.
We let $\pk{q}$ denote the set of all fractional edge packings for $q$.
The {\em fractional covering number $\tau^*$} is the maximum sum of weights over all possible
edge packings, $\tau^*(q) = \max_{\bu \in \pk{q}} \sum_j u_j$.

The {\em fractional edge cover} associates a non-negative weight $w_j$ to
each atom $S_j$, such that for every variable $x_i$, the sum of the weights of
the atoms that contain $x_i$ is at least 1. The {\em fractional edge cover number $\rho^*$}
is the minimum sum of weights over all possible fractional edge covers.
The notion of the fractional edge cover has
been used in the literature~\cite{DBLP:conf/focs/AtseriasGM08, NPRR12}
to provide lower bounds on the worst-case output size of a query (and consequently
the running time of join processing algorithms).

For any $\bx \subseteq \vars{q}$, we define the {\em residual query} $q_\bx$ as the query obtained from $q$ by removing
all variables $\bx$, and decreasing the arity of each relation accordingly (if the arity
becomes zero we simply remove the relation). For example, for the {\em triangle query}
$q(x,y,z) = R(x,y), S(y,z), T(z,x)$,
the residual query $q_{\{x\}}$ is $q_{\{x\}}(y,z) = R(y), S(y,z), T(z)$. Similarly, $q_{\{x,y\}}(z) =
S(z), T(z)$. Observe that every fractional edge packing of $q$ is
also a fractional edge packing of any residual query $q_\bx$, but the converse is not
true in general.  

We now define the {\em fractional edge quasi-packing} to be any edge packing of a residual 
query $q_\bx$ of $q$, where the atoms that have only
variables in $\bx$ get a weight of 0. Denote by
$\qpk{q}$ the set of all edge quasi-packings.
It is straightforward to see that $\pk{q} \subseteq \qpk{q}$; in other words, any packing is
a quasi-packing as well. The converse is
not true, since for example $(1,1,0)$ is a quasi-packing for the triangle query,
but not a packing. The {\em edge quasi-packing number} $\psi^*$ is the maximum sum
of weights over all edge quasi-packings:
$$\psi^*(q) =  \max_{\bu \in \qpk{q}} \sum_j u_j = \max_{\bx \subseteq \vars{q}} \max_{\bu \in \pk{q_\bx}} \sum_j u_j $$

\subsection{Previous Results}

Suppose that we are given a full CQ $q$, and input such that relation $S_j$ has size 
$M_j$ in bits (we use $m_j$ for the number of tuples). 
Let $\bM = (M_1, \dots, M_\ell)$ be the vector of the relation sizes.
For a given fractional edge packing $\bu \in \pk{q}$, we define as in~\cite{BKS14}:
\begin{align} 
  \quad \quad 
  L(\bu, \bM,p)   =  \left( \frac{ \prod_{j=1}^\ell M_j^{u_j}}{p} \right)^{1/\sum_{j=1}^\ell u_j} \label{eq:lx} 
\end{align}

Let us also define $L^{(q)}(\bM, p)  = \max_{\bu \in \pk{q}} L(\bu, \bM, p) $.
In our previous work~\cite{BKS14}, we showed that any algorithm that computes 
$q$ in a single round with $p$ servers must have load
$L \geq L^{(q)}(\bM, p)$. The instances used to prove this lower bound is
the class of {\em matching databases}, which are instances where each value appears
exactly once in the attribute of each relation. 
Hence, the above lower bound is not necessarily tight; indeed, as we will
see in the next section, careful choice of skewed input instances can lead to a stronger 
lower bound. \\

\noindent {\bf The HyperCube algorithm.}
To compute conjunctive queries in the \mpc\ model, we use the basic primitive
of the {\em HyperCube (HC) algorithm}.  The algorithm was
first introduced by Afrati and Ullman~\cite{DBLP:conf/edbt/AfratiU10},
and was later called the {\em shares} algorithm;  we use the name HC
to refer to the algorithm with a particular choice of shares.
The HC algorithm initially assigns to each variable $x_i$
 a {\em share} $p_i$, such that $\prod_{i=1}^k p_i =
p$. Each server is then represented by a distinct point $\mathbf{y}
\in \mathcal{P}$, where $\mathcal{P} = [p_1] \times \dots \times
[p_k]$; in other words, servers are mapped into a $k$-dimensional
hypercube. The HC algorithm then uses $k$ independently chosen hash
functions $h_i: \set{1,\dots,n} \rightarrow \set{1, \dots, p_i}$ (where $n$ is the domain size)
and sends each tuple $t$ of
relation $S_j$ to all servers in the destination subcube of $t$:
$$
 \mathcal{D}(t) = \setof{\mathbf{y} \in \mathcal{P}}
 {\forall x_i \in \vars{S_j}:  h_{i}(t[x_{i}]) = \mathbf{y}_{i}}
$$
where $t[x_i]$ denotes the value of tuple $t$ at the position of the variable $x_i$.
After the tuples are received,  each server locally computes $q$ for the subset of the input that it has received.  

If the input data has no skew, the above vanilla version of the HC algorithm is optimal 
for a single round. The lemma
below presents the specific conditions that define skew, and will be frequently used
throughout the paper.  

\begin{lemma}[Load Analysis for HC~\cite{BKS14}]
\label{lem:hashing}
Let $\mathbf{p} = (p_1, \dots, p_k)$ be the optimal shares of the HC algorithm.
Suppose that for every relation $S_j$ and every tuple $t$ over the attributes $U \subseteq [a_j]$ we have 
that the frequency of $t$ in relation $S_j$ is $m_{S_j}(t) \leq  m_j / \prod_{i \in U}p_i$.
Then with high probability the maximum load per server  is $\tilde{O}(L^{(q)}(\bM,p))$.
\end{lemma}

\section{One-Round Algorithms}
\label{sec:one-step}

In this section, we present tight upper and lower bounds for the
worst-case load of one-round algorithms that compute
conjunctive queries. Thus, we identify the database instances for
which the behavior in a parallel setting is the worst possible. Surprisingly,
these instances are often different from the ones that provide a worst-case
running time in a non-parallel setting.  

As an example, consider the triangle query $C_3 = R(x,y), S(y,z), T(z,x)$, where
all relations have $m$ tuples (and $M$ in bits). It is known from~\cite{DBLP:conf/focs/AtseriasGM08}
that the class of inputs that will give a worst-case output size, and hence a worst-case
running time, is one where each relation is a $\sqrt{m} \times \sqrt{m}$ fully bipartite
graph. In this case, the output has $m^{3/2}$ tuples. The load
needed to compute $C_3$ on this input in a single round is $\Omega(M/p^{2/3})$,
and can be achieved by using the HyperCube algorithm~\cite{BKS13} with shares
$p^{1/3}$ for each variable. Now, consider
the instance where relations $R,T$ have a single value at variable $x$, 
which participates in all the
$m$ tuples in $R$ and $T$; $S$ is a matching relation with $m$ tuples. In this case,
the output has $m$ tuples (and so $M$ bits), 
and thus is smaller than the worst-case output. However, as we will
see next, we can show that any one-round algorithm that computes the triangle query
for the above input structure requires $\Omega(M/p^{1/2})$ maximum load.

\subsection{An Optimal Algorithm}

We present here a worst-case optimal one-step algorithm that computes a
conjunctive query $q$. Recall that the HC algorithm achieves an optimal load on data without
skew~\cite{BKS14}. In the presence of skew, we will distinguish different cases, 
and for each case we will apply a different parametrization of the HC algorithm, 
using different shares.

We say that a value $h$ in relation $S_j$ is a {\em heavy hitter} in $S_j$ if the frequency 
of this particular value  in $S_j$, denoted $m_{S_j}(h)$, is at least $m_j/p$, where $m_j$ is the number of tuples 
in the relation. Given an output tuple $t$, we say that $t$ is {\em heavy at variable $x_i$} if
the value $t[x_i]$ is a heavy hitter in at least one of the atoms that include variable $x_i$.

We can now classify each tuple $t$ in the output depending on the positions where $t$ is
heavy. In particular, for any $\bx \subseteq \vars{q}$, let $q^{[\bx]}(I)$ denote the subset of the 
output that includes only the output tuples that are heavy at exactly the variables in $\bx$.
Observe that the case $q^{[\emptyset]}(I)$ denotes the case where the tuples are light at all
variables; we know from an application of Lemma~\ref{lem:hashing} that this case can be handled by the standard HC algorithm.
For each of the remaining $2^k-1$ possible sets $\bx \subseteq \vars{q}$, 
we will run a different variation
of the HC algorithm with different shares, which will allow us to compute $q^{[\bx]}(I)$ with
the appropriate load. Our algorithm will compute all the partial answers in parallel for each
$\bx \subseteq \vars{q}$, and thus requires only a single round. 

The key idea is to apply the HC algorithm by giving a non-trivial share only to the variables
that are not in $\bx$; in other words, every variable in $\bx$ gets a share of $1$.
In particular, we will assign to the remaining variables the shares we would assign
if we would execute the HC algorithm for the residual query $q_\bx$. We will thus choose
the shares by assigning $p_i = p^{e_i}$ for each $x_i \in \bx$ 
and solving the following linear program:
\begin{align}
\text{minimize}   \quad & \lambda \nonumber \\
\text{subject to} \quad
                  &  \sum_{i: x_i \notin \bx} -e_i \geq -1 \nonumber \\
\quad \forall j \text{ s.t. } S_j \in \atoms{q_\bx}:  & \sum_{i: x_i \in \vars{S_j} \setminus \bx} e_i + \lambda \geq \mu_j \nonumber \\
\quad \forall i \text{ s.t. } x_i \notin \bx: & e_i \geq 0, \quad \lambda \geq 0 \label{eq:primal:lp}
\end{align}
For each variable $x_i \in \bx$, we set $e_i = 0$ and thus the share is $p_i =1$.
We next present the analysis of the load for the above algorithm.

\begin{theorem} \label{th:one:round:upper}
Any full conjunctive query $q$ with input relation sizes $\bM$ can be computed 
in the \mpc\ model in a single round using $p$ servers with maximum load
$$ L = \tilde{O} \left( \max_{\bx \subseteq \vars{q} } L^{(q_\bx)}(\bM, p)  \right) $$
\end{theorem}

\begin{proof}
Let us fix a set of variables $\bx \subseteq \vars{q}$; we will show that the load of the algorithm
that computes $q^{[\bx]}(I)$ is $ \tilde{O} (L^{(q_\bx)}(\bM,p))$. The upper bound then follows
from the fact that we are running in parallel algorithms for all partial answers.

Indeed, let us consider how each relation $S_j$ is distributed using the shares
assigned. We distinguish two cases. If an atom $S_j$ contains variables that
are only in $\bx$, then the whole relation will be broadcast to all the $p$ servers.
However, observe that the part of $S_j$ that contributes to $q^{[\bx]}(I)$ is of size
at most $p^{a_j}$, where $a_j$ is the arity of the relation.

Otherwise, we will show that for every tuple $J$ of values over variables 
$\bv \subseteq \vars{S_j}$, 
we have that the frequency of $J$ is at most $m_j /\prod_{i: x_i \in \bv} p_i$.
Indeed, if $\bv$ contains only variables from $\bx$, then by construction
$\prod_{i: x_i \in \bv} p_i = 1$; we observe then that the frequency is
always at most $m_j$. If $\bv$ contains some variable $x_i \in \bv \setminus \bx$,
then the tuple $J$ contains at position $x_i$ a value that appears at most $m_j/p$
times in relation $S_j$, and since $\prod_{i: x_i \in \bv} p_i \leq p$ the claim holds. 
We can now apply Lemma~\ref{lem:hashing} to 
obtain that for relation $S_j$, the load will be 
$\tilde{O}(M_j/(\prod_{i: x_i \in \vars{S_j} \setminus \bx} p_i) )$. 
Summing over all atoms in the residual query $q_\bx$, and assuming that 
$m_j \gg p$ (and in particular that $p^{a_j}$ is always much smaller than the load), we
obtain that the load will be 
$\tilde{O}( \max_{j: S_j \in \atoms{q_\bx}} M_j/(\prod_{i: x_i \in \vars{S_j} \setminus \bx} p_i) )$,
which by an LP duality argument is equal to $\tilde{O} (L^{(q_\bx)}(\bM,p))$.
\end{proof}

When all relation sizes are equal, that is, $M_1 = M_2 = \dots = M_\ell = M$, the formula 
for the maximum load becomes $\tilde{O}(M/p^{1 / \psi^*(q)})$, where $\psi^*(q)$ is the 
edge quasi-packing number, which we have defined as
$\psi^*(q) = \max_{\bx \subseteq \vars{q}} \max_{\bu \in \pk{q_\bx}} \sum_j u_j$.
We will discuss about the quantity $\psi^*(q)$ in detail in Section~\ref{subsec:discussion}.
We will see next how the above algorithm applies to the triangle query $C_3$. 

\begin{example}
We will describe first how the algorithm works when each relation has size $M$
(and $m$ tuples). There are three different share
allocations, for each choice of heavy variables (all other cases are symmetrical).
\begin{description}
\item[$\bx = \emptyset$]: we consider only tuples with values of frequency 
$\leq m/p$. The HC algorithm will assign a share of $p^{1/3}$ to each variable, and the  maximum load will be $\tilde{O}(M/p^{2/3})$.
\item[$\bx = \{ x \}$]: the tuples have a heavy hitter value at variable $x$,
either in relation $R$ or $T$ or in both. The algorithm will give a share of 1 to $x$, and shares of
$p^{1/2}$ to $y$ and $z$. The load will be $\tilde{O}(M/p^{1/2})$.
\item[$\bx = \{ x,y\}$]: both $x$ and $y$ are heavy. In this case we broadcast 
the relation $R(x,y)$, which will have size at most $p^2$, 
and assign a share of $p$ to $z$. The load will be $\tilde{O}(M/p)$. 
\end{description}

Notice finally that the case where $\bx = \{x,y,z \}$ can be handled by broadcasting
all necessary information. The load of the algorithm is the maximum of the above quantities,
which is $\tilde{O}(M/p^{1/2})$.
When the size vector is $\bM = (M_1, M_2, M_3)$, the load achieved by the algorithm 
becomes $\tilde{O}(L)$, where:
$$ L = \max \left\{ \frac{M_1}{p}, \frac{M_2}{p}, \frac{M_3}{p},
\sqrt{\frac{M_1 M_2}{p}}, \sqrt{\frac{M_2 M_3}{p}}, \sqrt{\frac{M_1 M_3}{p}} \right\}.$$
\end{example}

\subsection{Lower Bounds}

We present here a worst-case lower bound for the load of one-step algorithms
for computing conjunctive queries in the
\mpc\ model, when the information known is the cardinality statistics 
$\bM = (M_1, \dots, M_\ell)$. The lower bound matches the upper bound in the
previous section, hence proving that the one-round algorithm is worst-case optimal.
We give a self-contained proof of the result in~\autoref{sec:appendix}, but many
of the techniques used can be found in previous work~\cite{BKS13, BKS14},
where we proved lower bounds for skew-free data and for input data with known
information about the heavy hitters.

\begin{theorem} 
\label{th:lower:skew} 
Fix cardinality statistics $\bM$ for a full conjunctive query $q$.
Consider any deterministic MPC algorithm that runs in one communication round on
$p$ servers and has maximum load $L$ in bits.  Then, for any $\bx \subseteq \vars{q}$,
there exists a family of (random) instances for which the load $L$ will be:
   $$L \geq \min_j \frac{1}{4 a_j} \cdot L^{(q_\bx)}(\bM, p).$$
\end{theorem}

Since $a_j \geq 1$, \autoref{th:lower:skew} implies that for any query $q$ there exists a 
family of instances such that  any one-round algorithm that computes $q$ must have load 
$\Omega(\max_{\bx \subseteq \vars{q}} L^{(q_\bx)}(\bM,p) )$.

\subsection{Discussion}
\label{subsec:discussion}

We present here the computation of the load of the one-round algorithm for
various classes of queries, and also discuss the edge quasi-packing number $\psi^*(q)$
and its connection with other query-related quantities.
 
Recall that we showed that when all relation sizes are equal to $M$, the load achieved is
of the form $\tilde{O}(M/p^{1/\psi^*(q)})$, where $\psi^*(q)$ is the quantity that maximizes the 
sum of the weights of the edge quasi-packing. The quantity $\psi^*(q)$ is
in general different from both the fractional covering number $\tau^*(q)$, and from the 
fractional edge cover number $\rho^*(q)$. Indeed, for the triangle query $C_3$ we have that
$\rho^*(C_3) =  \tau^*(C_3) = 3/2$, while $\psi^*(C_3) = 2$. 
Here we should remind the reader that $\tau^*$ 
describes the load for one-round algorithms on data without skew, which is 
$O(M/p^{1/\tau^*(q)})$. Also, $\rho^*$ characterizes the maximum possible output of a 
query $q$, which is $M^{\rho^*(q)}$.
We can show the following relation between the three quantities:

\begin{lemma}
For every conjunctive query $q$, $ \psi^*(q) \geq \max \{ \tau^*(q), \rho^*(q)\}$.
\end{lemma}

\begin{proof}
Since any edge packing is also an edge quasi-packing, it is straightforward to 
see that $\tau^*(q) \leq \psi^*(q)$ for every conjunctive query $q$.

To show that $\rho^*(q) \leq \psi^*(q)$, consider the optimal (minimum) edge cover
$\bu$; we will show that this is also an edge quasi-packing. First, observe that for
every atom $S_j$ in $q$, there must exist at least one variable $x \in \vars{S_j}$ such that
$\sum_{j: x \in \vars{S_j}} u_j =1 $. Indeed, suppose that for every variable in $S_j$ we have
that the sum of the weights strictly exceeds 1; then, we can obtain a better (smaller) edge cover
by slightly decreasing $u_j$, which is a contradiction. 

Now, let $\bx$ be the set of variables such that their cover in $\bu$ strictly exceeds 1, and
consider the residual query $q_\bx$. By our previous claim, every relation in $q$ is still
present in $q_\bx$, since every relation includes a variable with cover exactly one. Further,
for every variable $x \in \vars{q_\bx}$ we have $\sum_{j: x \in \vars{S_j}} u_j = 1$, and hence 
$\bu \in \pk{q_\bx}$.
\end{proof}

Another useful observation regarding the edge quasi-packing is the following lemma, which connects the edge quasi-packing of a query to the edge quasi-packing of any residual query.

\begin{lemma}
\label{lem:qpk:recursion}
For any conjunctive query $q$, $\psi^*(q) = \max_{\bx \subseteq \vars{q}}(\psi^*(q_\bx))$.
\end{lemma}

\begin{proof}
Let $\by \subseteq \vars{q_\bx}$ be the set of variables that produce the optimal edge quasi-packing for $q_\bx$, where $\bx \subseteq \vars{q}$. Then, $\psi^*(q_\bx) = \tau^*((q_\bx)_{\by}) = \tau^*(q_{\bx \cup \by}) \leq \psi^*(q)$. To prove the equality, let $q_{\by}$ be the residual query that produces the optimal edge quasi-packing for $q$. Then, for any $\bx \subseteq \by$, we have $\psi^*(q) = \tau^*(q_{\by}) = \tau^*((q_\bx)_{\by \setminus \bx}) \leq \psi^*(q_\bx)$.
\end{proof}

Observe that, if $\psi^*(q) > \tau^*(q)$, then we have to remove at least one variable from $q$ to obtain an optimal quasi-packing. In this case, we can also write $\psi^*(q) = \max_{x \in \vars{q}}(\psi^*(q_x))$. In other words, to compute the optimal quasi-packing we have to find the variable that will produce the residual query with the maximum quasi-packing, which is a recursive process.

In~\autoref{tab:complexity} we have computed the quantities $\tau^*, \rho^*, \psi^*$ for
several interesting classes of conjunctive queries: the star query $T_k$, the spiked star query $SP_k$,
the cycle query $C_k$, the line query $L_k$, the Loomis-Whitney join $LW_k$, the generalized
semi-join query $W_k$ and the clique (or full) query $K_k$. We next provide the detailed computation of these queries. \\

\begin{table}
  \centering
 { \tabulinesep=0.8mm
\begin{tabu}{|l|c|c|c|} \hline
  Conjunctive Query $q$ & $\tau^*(q)$    & $\rho^*(q)$ & $\psi^*(q)$ \\
    \hline

$T_k= \bigwedge_{j=1}^k S_j(z,x_j) $  & $1$ & $k$  & $k$  \\ \hline
$SP_k = \bigwedge_{i=1}^k R_i(z,x_i), S_i(x_i, y_i)$ & $k$ & $k+1$  & $k+1$ \\ \hline
$ K_k = \bigwedge_{1\le i<j \le k} S_{i,j}(x_i, x_j)$ & $k/2$ & $k/2$ & $k-1$ \\ \hline
$W_k = R(x_1, \dots, x_k) \bigwedge_{j=1}^k S_j(x_j) $ & $k$ & $1$ & $k$ \\ \hline
$L_k=  \bigwedge_{j=1}^k S_j(x_{j-1},x_j)$ & $\lceil k/2 \rceil$ & $\lceil (k+1)/2 \rceil$  
& $\lceil 2k/3 \rceil$ \\ \hline
$L_k^*=  R(x_0) \bigwedge_{j=1}^k S_j(x_{j-1},x_j)$ & $\lceil k/2 \rceil$ & $\lceil (k+1)/2 \rceil$  
& $\lceil (2k+1)/3 \rceil$ \\ \hline
$L_k^\dag=  R(x_0) \bigwedge_{j=1}^k S_j(x_{j-1},x_j) \bigwedge S(x_k)$ & $\lceil (k+1)/2 \rceil$ & $\lceil (k+1)/2 \rceil$  
& $\lceil (2k+2)/3 \rceil$ \\ \hline
$C_k = \bigwedge_{j=1}^{k} S_j(x_j,x_{(j \bmod k)+1})$ & $k/2$ & $k/2$ 
& $\lceil 2(k-1)/3 \rceil $  \\ \hline
$LW_k = \bigwedge_{I \subseteq [k], |I|=k-1} S_{I}(\bar x_{I})$ & $k/(k-1)$ &  $k/(k-1)$  & 2  \\ \hline
\end{tabu} }

  \caption{Computing the optimal edge packing $\tau^*$, edge cover $\rho^*$ and 
  edge quasi-packing $\psi^*$ for several classes of conjunctive queries.}
  \label{tab:complexity}
\end{table}

\noindent {\bf Star Queries.} Consider the star query 
$$T_k = S_1(z, x_1), S_2(z,x_2), \dots, S_k(z,x_k)$$
which generalizes the simple join $R(x,y), S(y,z)$ between two relations. It is easy to observe that the 
optimal edge packing does not exceed 1, since every relation includes the 
variable $z$. To obtain the maximum edge quasi-packing, we consider the residual
query $q_z = S_1(x_1), S_2(x_2), \dots, S_k(x_k)$ that removes the variable $z$, which is common in all 
atoms. Then, we can pack each relation $S_i$ with weight one, thus achieving a sum of $k$. As a 
consequence of this example, we obtain that the quantities $\tau^*$ and
$\psi^*$ generally are not within a constant factor. \\

\noindent {\bf Spiked Star Queries.} The optimal edge packing for $SP_k$ assigns a weight of $1$ to
each atom $S_i$; hence, $\tau^*(SP_k) = k$. To obtain the optimal quasi-packing, we consider the residual query $(SP_k)_{\{x_1, \dots, x_k \}} = R_1(z), \dots, R_k(z), S_1(y_1), \dots, S_k(y_k)$. This residual query admits an edge quasi-packing of weight $k+1$, by assigning a weight of 1 to any $R_i$ and a weight of 1 to every $S_i$. \\

\noindent {\bf Clique Queries.} 
Consider the clique query $K_k$, which includes all possible binary relations among the
$k$ variables $x_1, \dots, x_k$. The optimal edge packing is achieved by assigning a 
weight of $1/(k-1)$ to each relation, and thus $\tau^*(K_k) = \binom{k}{2} \frac{1}{k-1} = k/2$.
The corresponding share allocation for the HC algorithm assigns an equal share of $p^{1/k}$ to each variable. For the edge quasi-packing, consider the residual query
$(K_k)_{x_1}$, and notice that it includes $(k-1)$ unary relations, one for each of variable
$x_2, \dots, x_k$. Hence, we can obtain an edge packing of total weight $k-1$ by assigning a weight of $1$ to each. This is optimal, since by symmetry we have that for any $x_i$, $\psi^*(K_k) = \psi^*((K_k)_{x_i})$, and the quasi packing of the residual query $(K_k)_{x_i}$ can be at most $k-1$, since it has $k-1$ variables. \\

\noindent {\bf Generalized Semi-Join Queries.} \
The optimal edge packing assigns a weight of $1$ to each relation $S_j$, thus achieving a total weight of $k$. Since the edge quasi-packing cannot exceed the number of variables, $k$ is also the optimal value for the edge quasi-packing for $W_k$.  \\

\noindent {\bf Line and Cycle Queries.} To compute the edge quasi-packing for the cycle and line query, we need to consider two more CQ classes, which we define next:
\begin{align*}
L_k &=  S_1(x_0, x_1), S_2(x_1, x_2), \dots, S_k(x_{k-1}, x_k) \\
L_k^{*} & = R(x_0), S_1(x_0, x_1), S_2(x_1, x_2), \dots, S_k(x_{k-1}, x_k) \\
L_k^{\dag}  & = R(x_0),  S_1(x_0, x_1), S_2(x_1, x_2), \dots, S_k(x_{k-1}, x_k), T(x_k)
\end{align*}
The query $L_k^*$ adds to the line query a unary relation to the left, and the query $L_k^{\dag}$ adds a unary relation to the left and right of the query.

\begin{lemma}
\label{lem:line:lr}
$\psi^*(L_k^{\dag}) = \lceil (2k+2)/3\rceil$. 
\end{lemma}

\begin{proof}
We will prove this using induction. For the base of the induction, notice that $\psi^*(L_0^{\dag}) = 1 $  and $\psi^*(L_1^{\dag}) = 2 $. For $k  \geq 2$, it is easy to see that removing one variable that is not $x_0$ or $x_k$ can never decrease the packing, since we remove one constraint without removing any relation. Thus, 
\begin{align*}
\psi^*(L_k^{\dag}) & = \max_{i=1}^{k-1} \psi^*({(L_k^{\dag})_{x_i}}) \\
& = \max_{i=1}^{k-1} \{\psi^{*}(L_{i-1}^\dag) + \psi^{*}(L_{k-i-1}^\dag) \} \\
& = \max_{i=1}^{k-1} \{\lceil (2(i-1)+2)/3\rceil + \lceil (2(k-i-1)+2)/3\rceil \} \\
& = \max_{i=1}^{k-1} \{\lceil 2i/3\rceil + \lceil 2(k-i)/3\rceil \} 
\end{align*}
To compute the max, we consider three different cases, depending on the modulo 3 of $i$. If $i=3m$, then the quantity becomes $2m+\lceil 2k/3 -2m \rceil  = \lceil 2k/3 \rceil$. If $i = 3m+1$, then it becomes $\lceil 2m+2/3 \rceil + \lceil 2k/3 -2m -2/3\rceil  = \lceil 2k/3 + 1/3 \rceil $. Lastly, for $i = 3m+2$ it becomes $\lceil 2m+4/3 \rceil + \lceil 2k/3 -2m -4/3\rceil = \lceil (2k+2)/3 \rceil$. Thus, the maximum is achieved when we choose $x_2, x_5, x_8, \dots$ and it is $\lceil (2k+2)/3 \rceil$.
\end{proof}

As a corollary of \autoref{lem:line:lr}, we can compute the optimal edge quasi-packing of the cycle query
$$C_k = S_1(x_1, x_2), S_2(x_2, x_3), \dots, S_k(x_k, x_1).$$ 
Observe that by removing a single variable, we obtain the residual query $L_{k-2}^\dag$, which will have at least the same value as $C_k$, since we have removed a constraint (variable) but no relations. Hence, $\psi^*(C_k) = \psi^*(L_{k-2}^\dag) = \lceil (2(k-2)+2)/3\rceil = \lceil 2(k-1)/3\rceil$.

\begin{lemma}
\label{lem:line:l}
$\psi^*(L_k^{*}) = \lceil (2k+1)/3\rceil$. 
\end{lemma}

\begin{proof}
We prove this also by induction. For the base case, we have $\psi^*(L_0^{*}) = 1 $  and $\psi^*(L_1^{*}) = 1 $. For $k \geq 2$, removing any variable that is not $x_0$ will not decrease the optimal edge quasi-packing, hence:
\begin{align*}
\psi^*(L_k^{*}) & = \max_{i=1}^{k} \psi^*({(L_k^{*})_{x_i}}) \\
& = \max_{i=1}^{k} \{\psi^{*}(L_{i-1}^\dag) + \psi^{*}(L_{k-i-1}^*) \} \\
& = \max_{i=1}^{k} \{\lceil 2i/3\rceil + \lceil (2(k-i)-1)/3\rceil \} 
\end{align*}
Using the same reasoning as the proof of~\autoref{lem:line:lr}, we can see that the quantity is maximized for $i = 3m+2$, and then $\psi^*(L_k^{*}) = \lceil (2k+1)/3 \rceil$.
\end{proof}

We can finally compute the optimal edge quasi-packing for $L_k$. We will show that $\psi^*(L_k) = \lceil 2k/3 \rceil$ using induction. The prove is similar to the proof of~\autoref{lem:line:lr}. For the base of the induction, notice that $\psi^*(L_1) = 1$. For any $k \geq 2$, we have that $\psi^*(L_k) =  \max_{i=1}^{k-1} \{\psi^*(L_{i-1}^*) + \psi^*(L_{k-i+1}^*)\} = \max_{i=1}^{k} \{\lceil (2i-1)/3\rceil + \lceil (2(k-i)-1)/3\rceil \} $. The latter quantity is maximized at every $i = 3m+1$, in which case $\psi^*(L_k) = \lceil 2k/3 \rceil$. \\

\noindent {\bf Loomis-Whitney Joins.} To optimal edge packing for LW joins assigns a weight of $1/(k-1)$ to each of the $k$ relations, which implies a total weight of $k/(k-1)$. To compute the optimal edge quasi-packing, observe that by removing any one variable (w.k.o.g. this can be $x_k$), the residual query $(LW_k)_{x_k}$ is equivalent to $LW_{k-1} \bigwedge S'(x_1, \dots, x_{k-1})$. Since $S'$ contains all variables, it is easy to see that $\psi^*((LW_{k})_{x_k}) = \psi^*(LW_{k-1})$. Hence, $\psi^*(LW_k) = \max \{\tau^*(LW_k), \psi^*(LW_{k-1})\}$, and by repeating the argument we have $\psi^*(LW_k) = \max_{k=2}^k \{k/(k-1)\} = 2$.

\section{Multi-round Algorithms}
\label{sec:multi-step}

In this section, we present algorithms for multi-round computation of several conjunctive queries in the case where the relation sizes are all equal. We also prove a lower bound that
proves that their optimality (up to a polylogarithmic factor).

\subsection{Multi-round Lower Bound}

We prove here a general lower bound for any algorithm that computes conjunctive
queries using a constant number of rounds. Observe that the lower bound is expressed in
terms of number of tuples (and not bits); our upper bounds will be expressed in terms of bits,
and thus will be a $\log(n)$ factor away from the lower bound, where $n$ is the domain size.

\begin{theorem}
Let $q$ be a conjunctive query. Then, there exists a family of instances where relations
have the same size $M$ in bits (and $m$ in tuples) such that every
algorithm that computes $q$ with $p$ servers using a constant number of rounds 
requires load  $\Omega(m/p^{1/\rho^*(q)})$.
\end{theorem}

\begin{proof} 
In order to prove the lower bound, we will use a family of instances that give the
maximum possible output when every input relation has at most $m$ tuples, 
which is $m^{\rho^*(q)}$ (see~\cite{DBLP:conf/focs/AtseriasGM08}).
We also know how we can construct such a worst-case instance: for each variable
$x_i$ we assign an integer $n_i$ (which corresponds to the domain size of the variable), and
we define each relation as the cartesian product of the domains of the variables it
includes: $\times_{i: x_i \in \vars{S_j}} [n_i]$. The output size then will be 
$\prod_{i} n_i = m^{\rho^*(q)}$ (using a LP duality argument).

We now define the following random instance $I$ as input for the query $q$: 
for each relation $S_j$, we choose each tuple from the full cartesian product of the 
domains independently at random with probability $1/2$. It is straightforward to see
that the expected size of the output is $E[|q(I)|] = (1/2)^\beta \prod_i n_i$, where $\beta$
is the maximum number of relations where any variable occurs (and thus a constant
depending on the query). Using Chernoff's bound we can claim an even stronger result:
the output size will be $\Theta(m^{\rho^*(q)})$ with high probability (the failure probability is exponentially small in $m$).

Now, assume that algorithm $\mA$ computes $q$ with load $L$ (in bits) in $r$ rounds. 
Then, each server receives at most $L' = r\cdot L$ bits. Fix some server and let $\msg$ be 
the whole sequence of bits received by this server during the computation; hence, 
$|\msg| \leq L'$. We will next compute how many
tuples from $S_j$ are known by the server, denoted $K_\msg(S_j)$. W.l.o.g. we can
assume that all $L'$ bits of $\msg$ contain information from relation $S_j$.

We will show that the probability of the event $K_\msg(S_j) > (1+\delta) L'$ is exponentially
small on $\delta$. Let $m_j = \prod_{i: x_i \in \vars{S_j}} n_i \leq m$. 
Observe first that the total number of message configurations of size $L'$ is at most
$2^{L'}$. Also, since the size of the full cartesian product is $m_j$,  $\msg$ can
encode at most $2^{m_j-(1+\delta) L'}$ relations $S_j$ (if $m_j < (1+\delta) L'$, then trivially
the probability of the event is zero, and $S_j$ will have "few" tuples).
It follows that 
$$P(K_\msg > (1+\delta) L') < 2^{L'} \cdot 2^{m_j-(1+\delta)L'} \cdot (1/2)^{m_j} =(1/2)^{\delta L'} $$

So far we have shown that with high probability each server knows at most $L'$ tuples from
each relation $S_j$, and further that the total number of output tuples is $\Theta(m^{\rho^*(q)})$.
However, if a server knows $L'$ tuples from each relation, using the AGM bound 
from~\cite{DBLP:conf/focs/AtseriasGM08}, it can output at most $(rL)^{\rho^*(q)}$ tuples. The
result follows by summing over the output of all $p$ servers, and using the fact that the 
algorithm has only a constant number of rounds.
\end{proof}

The theorem implies that whenever $\psi^*(q) = \rho^*(q)$ the one-round algorithm is 
essentially worst-case optimal, and using more rounds will not result in an algorithm with 
better load. As a result,
and following our discussion in the previous section, the classes of queries 
$T_k$ and $SP_k$ can be optimally computed in a single round. 
This may seem counterintuitive, but recall
that we study worst-case optimal algorithms; there may be instances where using
more rounds is desirable, but our goal is to match the load for the worst such instance.

We will next present algorithms that match (within a polylogarithmic factor) the above lower bound using strictly more
than one round. We start with the algorithm for the triangle query $C_3$, in order to
demonstrate our novel technique and prove a key result (Lemma~\ref{lem:join:1skew}) 
that we will use later in the section.

\subsection{Warmup: Computing Triangles in 2 Rounds}

The main component of the algorithm that computes triangles is a parallel algorithm
that computes the join $S_1(x,z),S_2(y,z)$ in a single round, for the case where skew 
appears exclusively in one of the two relations. If the relations have size $M_1, M_2$
respectively, then we have shown that the load can be as large as $\sqrt{M_1 M_2 / p}$. 
However, in the case of one-sided skew, we can compute the join with maximum load only 
$\tilde{O}(\max \{ M_1, M_2 \}/p)$.

\begin{lemma} 
\label{lem:join:1skew}
Let $q = S_1(x, z),S_2(y, z)$, and let $m_1$ and $m_2$ be the relation sizes (in tuples) of
$S_1, S_2$ respectively. Let  $m = \max\{m_1, m_2\}$. If the degree of every value of the
variable $z$ in $S_1$, $m_{S_1}(z)$, is at most $m/p$, then we can compute $q$ in a 
single round with $p$ servers and load (in bits) $\tilde{O}(M/p)$, where $M = 2 m \log(n)$
($n$ is the domain size).
\end{lemma}

\begin{proof}
We say that a value $h$ is a heavy hitter in $S_2$ if the degree  of $h$ in $S_2$ is 
$m_{S_2}(h) > m/p$. By our assumption, there are no heavy hitters
in relation $S_1$.

For the values $h$ that are not heavy hitters in $S_2$, we can compute the join by applying
the standard HC algorithm (which is a hash-join that assigns a share of $p$ to $z$); 
the load analysis of Lemma~\ref{lem:hashing} will give us a 
load of $\tilde{O}(M/p)$ with high probability. 

For every heavy hitter $h$, the algorithm computes the subquery 
$q[h / z] = S_1(x,h), S_2(y,h)$, which is equivalent to computing the residual query 
$q_z = S_1'(x), S_2'(y)$, where $S_1'(x) = S_1(x, h)$ and $S_2'(y) = S_2(y,h)$.
We know that $|S_2'| = m_{S_2}(h)$ and $|S_1'| \leq m/p$ by our assumption.
The algorithm now allocates $p_h = \lceil p \cdot m_{S_2}(h) /m \rceil$ exclusive servers
 to compute  $q[h/z]$ for each heavy hitter $h$. To compute $q[h/z]$ with $p_h$ servers,
 we simply use the simple broadcast join that assigns a share of $p$ to variable $x$ and 
 $1$ to $y$. A simple
 analysis will give us that the load (in tuples) for each heavy hitter $h$ is 
$$\tilde{O} \left( \frac{|S_2'|}{ p_h} + |S_1'| \right) = 
\tilde{O} \left( \frac{m_{S_2}(h)}{p \cdot m_{S_2}(h) / m}  + m/p) \right) = 
\tilde{O} (m/p)$$
Finally, observe that the total number of servers we need is $\sum_h p_h \leq 2p$, 
hence we have used an appropriate amount of the available $p$ servers.
\end{proof}

Thus, we can optimally compute joins in a single round  in the 
presence of one-sided skew. We can apply Lemma~\ref{lem:join:1skew} to obtain a
useful corollary for the {\em semi-join} query $q = R(z), S(y,z)$. Indeed, notice that
we can extend $R$ to a binary relation $R'(x,z)$, where $x$ is a dummy variable that
takes a single value; then, the semi-join becomes essentially a join, where $R'$ has no
skew, since the degree of $z$ in $R'$ will be always one. Consequently:

\begin{corollary}
\label{lem:semi-join}
Consider the semi-join query $q = R(z),S(y, z)$, and let $M_1$ and $M_2$ be the relation
sizes of $R, S$ respectively in bits. Then we can compute $q$ in a 
single round with $p$ servers and load $\tilde{O}(\max \{M_1, M_2 \}/p)$.
\end{corollary}

We now outline the algorithm for computing triangles using two rounds. The central idea
in the algorithm is to identify the values that create skew in the computation, and spread
this computation into more rounds.

\begin{theorem}
The triangle query $C_3 = S_1(x_1, x_2), S_2(x_2, x_3), S_3(x_3, x_1)$ 
on input with sizes $M_1 = M_2 = M_3 = M$ can be computed by
an \mpc\ algorithm in 2 rounds with $\tilde{O}(M/p^{2/3})$ load, under any input
data distribution.
\end{theorem}

\begin{proof}
We say that a value $h$ is heavy if for some relation $S_j$, we have $m_j(h) > m/p^{1/3}$. We 
first compute the answers for the tuples that are not heavy at any variable. 
Indeed, if for every value we have that the degree is at most $m/p^{1/3}$, then the load analysis 
(Lemma~\ref{lem:hashing})
tells us that we can compute the output in a single round with load $\tilde{O}(M/p^{2/3})$
using the HC algorithm that allocates a share of $p^{1/3}$ to each variable.

Thus, it remains to output the tuples for which at least one variable has a heavy value.
Without loss of generality, consider the case where variable $x_1$ has heavy values
and observe that there are at most $2 p^{1/3}$ such heavy values for $x_1$ ($p^{1/3}$ for $S_1$ and $p^{1/3}$ for $S_3$). 
For each heavy value $h$, we assign an {\em exclusive} set of $p' = p^{2/3}$ servers
in order to compute the query $q[h/x_1] = S_1(h, x_2), S_2(x_2, x_3), S_3(x_3, x_1)$, which 
is equivalent to computing the residual query 
$$q' = S_1'(x_2), S_2(x_2, x_3), S_3'(x_3).$$

To compute $q'$ with $p'$ servers, we use 2 rounds. In the first round, we compute in
parallel the semi-join queries $S_{12}(x_2, x_3) = S_1'(x_2), S_2(x_2, x_3)$ and 
$S_{23}(x_2, x_3) = S_2(x_2, x_3), S_3'(x_3)$.
Since $|S_1'| \leq m$ and $|S_2'| \leq m$, we can apply Corollary~\ref{lem:semi-join} for
semi-join computation to obtain that we can achieve this computation with load (in tuples)
$\tilde{O}(m/p') = \tilde{O}(m/p^{2/3})$.
Observe that the intermediate relations $S_{12}, S_{23}$ have size at most $m$.
In the second round, we simply perform the intersection of the relations
$S_{12}, S_{23}$; this can be achieved with 
tuple load $O(m/p') = O(m/p^{2/3})$. Observe that the load for computing the 
intersection of two or more relations does not have any additional logarithmic factors.
\end{proof}

Notice that the 2-round algorithm achieves a better load than the 1-round
algorithm in the worst-case scenario. Indeed, in the previous section we proved that there
 exist instances
for which we can not achieve load better than $O(M/p^{1/2})$ in a single round. By using an
additional round, we can beat this bound and achieve a better load. This confirms our intuition
that with more rounds we can reduce the maximum load. Moreover,  observe that the load
achieved matches the multi-round lower bound (within a polylogarithmic factor).

\subsection{Computing General Conjunctive Queries}

We now generalize the ideas of the above example, and extend our results to several
standard classes of conjunctive queries. Throughout this section, we assume that all relations have the same size $M$ in bits (and $m$ in tuples).  
We present in detail multiround algorithms for the line query $L_k$, the cycle query $C_k$, the Loomis-Whitney join $LW_k$ and the clique (full) query $K_k$. 

\subsubsection{Line Queries}

We start by studying the line query 
$$L_k = S_1(x_0, x_1), S_2(x_1, x_2), \dots, S_k(x_{k-1}, x_{k}).$$

\begin{lemma}\label{lem:line}
The line query  $L_k$ can be computed by an \mpc\ algorithm in at most $\lfloor k/2 \rfloor$ rounds with load $\tilde{O}(M/p^{1/\lceil (k+1)/2 \rceil})$ for any $k \geq 2$.
\end{lemma}

We will describe the algorithm using induction on the length $k$ of the query $L_k$. For the base case, observe that for any $k \leq 4$, $\psi^*(L_k) = \rho^*(L_k)$ and thus the one-round algorithm from the previous section is optimal. Suppose now that we want to compute the query $L_k$ for $k \geq 5$. We now distinguish two cases, depending on whether $k$ is odd or even. \\

\noindent \textbf{Even Length.} Let $k=2n$. In this case, we will show that we can achieve a load of $\tilde{O}(M/p^{1/(n+1)})$. The algorithm, instead of computing $L_k$ directly, computes the cartesian product of $L_{k-1}$ with the single relation $S_k$. We assign $p_0 = p^{n/(n+1)}$ servers to compute $L_{k-1}$ and $p_1 = p^{1/(n+1)}$ servers for $S_k$. Notice that $p = p_0 \cdot  p_1$, and thus the cartesian product can indeed be computed using $p$ servers. By the induction hypothesis, to compute $L_{k-1}$ with $p_0$ servers we need $\lfloor (k-1)/2 \rfloor \leq \lfloor k/2 \rfloor$ rounds and load $\tilde{O}(M/p_0^{1/\lceil k/2 \rceil})$, where $p_0^{1/\lceil k/2 \rceil} = p^{(1/n) \cdot n/(n+1)} = p^{1/(n+1)}$. We can further distribute $S_k$ (in one round) such that the load is $O(M/p_1) = O(M/p^{1/(n+1)})$. \\

\noindent \textbf{Odd Length.} Let $k=2n-1$. In this case, we will show that we can achieve a load of $\tilde{O}(M/p^{1/n})$. We say that the variable $x_1$ is heavy if its degree in $S_0$ is at least $m/p^{1/n}$. We distinguish two cases:
\begin{enumerate}
\item {\em Light $x_1$:} We compute the cartesian product of $q_0 = S_3(x_2, x_3), \dots S_k(x_{k-1}, x_k)$ with the join $q_1 = S_1(x_0,x_1), S_2(x_1, x_2)$. We assign $p_0 = p^{(n-1)/n}$ servers to $q_0$ and $p_1 = p^{1/n}$ servers to $q_1$. Observe again that $p = p_0 \cdot p_1$. Since $q_0$ is isomorphic to $L_{k-2}$, by the induction hypothesis we can compute $q_0$ using $\lfloor (k-2)/2 \rfloor \leq \lfloor k/2 \rfloor$ rounds with load $\tilde{O}(M/p_0^{1/\lceil (k-1)/2 \rceil})$, where $p_0^{1/\lceil (k-1)/2 \rceil} =  p^{1/n}$. As for $q_1$, notice that it is a join that we want to compute with $p_1$ servers. By~\autoref{lem:join:1skew} and the fact that $x_1$ is light, we can compute $q_1$ with load $\tilde{O}(M/p_1) = \tilde{O}(M/p^{1/n})$ in one round.
\item {\em Heavy $x_1$:} Let $d_1, \dots, d_r$ be the degrees of the heavy hitters in $S_1$: we have that $\sum_{i=1}^r d_i \leq m$. For each heavy hitter $h_i$ with degree $d_i$, we assign $p^{(i)} = p^{(i)}_0 \cdot p^{(i)}_1$ servers, where $p^{(i)}_0 = \frac{d_i}{m} \cdot p^{1/n}$ and  $p^{(i)}_1 =   p^{(n-1)/n}$. Notice that $\sum_{i=1}^r p^{(i)} \leq p$. To compute the query $L_k[h_i / x_1] = S_1(x_0, h_i), S_2(h_i, x_2), \dots,  S_k(x_{k-1}, x_{k})$, we compute the cartesian product of $q_0 = S_1'(x_0)$ and $q_1 = S_2'(x_2), \dots,  S_k(x_{k-1}, x_{k})$. We use $p^{(i)}_0$ servers to distribute $S_1'$ and achieve a tuple load of $O(d_i/p^{(i)}_0) = O(m/p^{1/n})$. We use $p_1^{(i)}$ servers to compute $q_1$: in the first round, we compute the semijoin $S_2'(x_2), S_3(x_2, x_3)$, and we are left then with $L_{k-2}$, which by the induction hypothesis we can compute in $\tilde{O}(M/(p^{(i)}_1)^{1/\lceil (k-1)/2 \rceil}) = \tilde{O}(M/p^{1/n})$. The total number of rounds for this case is $1 + \lfloor (k-2)/2 \rfloor = \lfloor k/2 \rfloor$ rounds.
\end{enumerate}

\subsubsection{Cycle Queries}

We next consider the cycle query
$$ C_k = S_1(x_1, x_2), S_2(x_2, x_3), \dots, S_k(x_k,x_1)$$ 

\begin{lemma}\label{lem:cycle}
The cycle query  $C_k$ can be computed by an \mpc\ algorithm in at most $\lceil k/2 \rceil$ rounds with load $\tilde{O}(M/p^{2/k})$ for any $k \geq 3$.
\end{lemma}

The algorithm that computes $C_k$ will use as a component the algorithm that computes the line query $L_k$. As with the case of the line query, we will distinguish two cases, depending on whether $k$ is odd or even. \\

\noindent \textbf{Odd Length.} The algorithm is a generalization of the method for 
computing triangle queries presented as a warmup example. We say that a value $h$ is {\em heavy} for variable $x_i$ if for relation $S_{i-1}$ or $S_i$,
we have $m_i(h) > m/p^{1/k}$ or $m_{i-1}(h) > m/p^{1/k}$. We 
first compute the answers for the tuples that are not heavy at any position.
Lemma~\ref{lem:hashing} implies that we can compute the output in a single round with 
load $\tilde{O}(M/p^{2/k})$, by applying the vanilla HC algorithm for cycles, where each
variable has equal share $p^{1/k}$.

We next compute the tuples that are heavy at variable $x_1$ (we similarly do this for 
every variable $x_i$);  observe that there are at most $2p^{1/k}$ such values. 
For each such heavy value $h$, we will assign an exclusive number of 
$p' = p^{1-1/k}$ servers, such that the total number of servers we use is $(2p^{1/k}) \cdot p'
= \Theta(p)$, and using these servers we will compute the query 
$q[h/x_1] = S_1(h, x_2), \dots, S_k(x_k, h)$, which amounts to computing the residual query 
$q' = q_{x_1}$:
$$q' = S_1'(x_2), S_2(x_2, x_3), \dots, S_{k-1}(x_{k-1}, x_k), S_k'(x_k)$$

To compute $q'$ with $p'$ servers we need several rounds of computation. 
In round one, we compute in parallel the two semi-joins 
\begin{align*}
S_{1,2}(x_2, x_3) = S_1'(x_2), S_2(x_2, x_3), \quad \quad
S_{k,k-1}(x_{k-1}, x_k)  = S_{k-1}(x_{k-1}, x_k), S_k'(x_k)
\end{align*}
which can be achieved with tuple load $\tilde{O}(m/p') = \tilde{O}(m/p^{1-1/k})$, 
since $|S_1'| \leq m$ and $|S_k'| \leq m$ (by applying Corollary~\ref{lem:semi-join}).
Since for any $k \geq 3$ we have $1-1/k \geq 2/k$, the load for the first round will be
$\tilde{O}(M/p^{2/k})$. Next we compute the query
$$ q'' = S_{1,2}(x_2, x_3), S_3(x_3, x_4) ,\dots, S_{k-1}(x_{k-1}, x_k), S_{k,k-1}(x_{k-1}, x_k)$$
which is isomorphic to the line query $L_{k-2}$, where each relation has
size at most $m$. We know from ~\autoref{lem:line} that we can compute such a 
query with tuple load $\tilde{O}(m/p'^{1/ \lceil (k-1)/2 \rceil}) = \tilde{O}(m/p^{2/k})$ using $\lfloor (k-2)/2 \rfloor$ rounds. Thus we need $\lfloor k/2 \rfloor$ total rounds and load $\tilde{O}(M/p^{2/k})$. \\

\noindent \textbf{Even Length.} 
For even length cycles, our previous argument does not work, and we have to use a
different approach. We say that a value $h$ is $\delta$-heavy, for some $\delta \in [0,1]$,
if the degree of $h$ is at least $m/p^{\delta}$ in some relation.
We distinguish two different cases:

\begin{enumerate}
%
\item Suppose that there exist two variables $x_i, x_{i'}$ such that $(i-i')$ is an odd number, 
$x_i$ is $\delta$-heavy, $x_{i'}$ is $\delta'$-heavy, and $\delta+ \delta' \leq 2/k$. Observe
that there are at most $p^{\delta+ \delta'} \leq p^{2/k}$ such pairs of heavy values: for each
such pair, we assign $p' = p^{1-2/k}$ explicit servers to compute the residual query
$q' = (C_k)_{x_i, x_j}$ in two rounds. We now consider two subcases. If $i' = i+1$, then
$x_i, x_{i'}$ belong in the same relation $S_i$. Then, by performing the semi-join computations
in the first round, we reduce the computation of the next rounds to the residual query $L_{k-3}$, which requires 
tuple load $\tilde{O}(m/p'^{1/\lceil (k-2)/2 \rceil}) = \tilde{O}(m/p^{2/k})$, since $k$ is even. 
Otherwise, if $x_i,x_{i'}$ are not in the same relation, we still do the semi-joins in the first round,
and then notice that in the subsequent rounds we need to compute the cartesian product of two line queries, $L_{\alpha},
L_{\beta}$, where $\alpha + \beta = k-4$ and both are odd numbers. To perform this cartesian
product, we will split the $p'$ servers into a $p^{(\alpha+1)/k} \times p^{(\beta+1)/k}$ grid, and
within each row/column compute the line queries. Then, the tuple load will be
$\tilde{O}(m/p^{((\alpha+1)/k) \cdot (1/\lceil (\alpha+1)/2 \rceil)}) = 
\tilde{O}(m/p^{((\beta+1)/k) \cdot (1/\lceil (\beta+1)/2 \rceil)}) = \tilde{O}(m/p^{2/k})$. 

\item Otherwise, define $\delta_{even}$ as the largest number in $[0,1]$ such
that for every even variable the frequency is at most $m/p^{\delta_{even}}$.
Similarly define $\delta_{odd}$. Since we do not fall in the previous case, it
must be that $\delta_{even} + \delta_{odd} \geq 2/k$. W.l.o.g. assume that
$\delta_{even} \geq \delta_{odd}$. Then, consider the HC algorithm with the
following share allocation: for odd variables assign $p_o = p^{\delta_{odd}}$,
and for even variables assign $p_e = p^{2/k-\delta_{odd}}$. Since the odd
variables have degree at most $m/p^{\delta_{odd}}$, there are no skewed values
there. As for the even variables, their degree is at most $m/p^{\delta_{even}} 
\leq m/p^{2/k-\delta_{odd}} = m/p_e$. Hence, the tuple load achieved will be 
$\tilde{O}(m/(p_o p_e) = \tilde{O}(m/p^{2/k})$.
 In the case where $p_e$ is ill-defined
because $\delta_{odd} > 2/k$, we also have that $\delta_{even} > 2/k$ and in
this case we can just apply the standard HC algorithm that assigns a share of
$p^{1/k}$ to every variable.  
\end{enumerate}

\subsubsection{Loomis-Whitney Joins}

We show here how to compute the the Loomis-Whitney join
$$ LW_k = S_1(x_2, \dots, x_k), S_2(x_1, x_3, \dots, x_k), \dots, S_k(x_1, \dots, x_{k-1})$$
using a 2-round algorithm.  Notice that $LW_3$ is the triangle
query $C_3$; as we will see, the algorithmic idea here is the same as the one for
even cycles.

\begin{lemma}
The Loomis-Whitney join $LW_k$ can be computed by an MPC
algorithm with 2 rounds and load $\tilde{O}(M/p^{1-1/k})$ for any $k \geq 3$.
\end{lemma}

The algorithm sets the threshold for a heavy hitter to $p^{1/k}$. For the case of tuples where no values is heavy, the one-round HC algorithm that assigns a share of $p^{1/k}$ to each variable achieves the desired load of $\tilde{O}(M/p^{1-1/k})$.

Consider now the tuples where $x_1$ is heavy (we do this similarly for all variables). 
For each heavy hitter value $h$, we assign $p' = p^{1-1/k}$ explicit servers to compute
the residual query $q' = q_{x_1}$, where
$$ q_{x_1} = S_1(x_2, \dots, x_k) \bigwedge_{j=2}^k S_j'(x_2, \dots, x_{j-1}, x_{j+1}, \dots, x_k)$$
and the size of each relation is at most $m$. We will need 2 rounds to compute this query.
Since for every $j=2, \dots, k$ we have that
$\vars{S_j'} \subseteq \vars{S_1}$, in the first round we will compute for every $j \geq 2$ the
semijoins $S_j'' = S_1(x_2, \dots, x_k), S_j'(x_2, \dots, x_{j-1}, x_{j+1}, \dots, x_k)$, which we can
do with load $\tilde{O}(M/p') = \tilde{O}(M/p^{1-1/k})$ by applying Corollary~\ref{lem:semi-join}.
Notice that in this case we join on multiple variables, but the corollary still holds, since we can hash on all the joining variables. In the second round, we will compute the
intersection $\bigwedge_{j=2}^k S_j''(x_2, \dots, x_k)$, which can be achieved with load 
$O(M/p') = O(M/p^{1-1/k})$.

\subsubsection{Clique Queries} 

We next present an algorithm for the clique query
$$ K_k = \bigwedge_{1\le i<j \le k} S_{i,j}(x_i, x_j).$$
Again, notice that $K_3$ is the triangle query $C_3$; however, for this class of queries we will use
an algorithm that requires in general more than 2 rounds. It is not clear whether 2 rounds could be sufficient for achieving the required load.
We will prove the following result.

\begin{lemma}
The clique query $K_k$ can be computed by an MPC algorithm in $k-1$ rounds with load $\tilde{O}(M/p^{2/k})$
for any $k \geq 3$.
\end{lemma}

\begin{proof}
We will use induction on the parameter $k$.
We have already shown the result for the base case $k=3$ (the triangle query), where we need 2 rounds to compute
the triangle query $C_3$. Consider now $K_k$; we set the heaviness threshold as usual at
$p^{1/k}$. For the tuples that contain no heavy values, the HC algorithm that assigns a share of
$p^{1/k}$ to each variable guarantees the desired load of $\tilde{O}(M/p^{2/k})$.

Consider next w.l.o.g. the tuples for which variable $x_1$ is heavy. 
There are at most $(k-1)p^{1/k}$ such heavy values $h$,
and for each we assign exclusively $p' = p^{1-1/k}$ servers to compute the residual query 
$ q' = \bigwedge_{i=2}^k S_{1,i}'(x_i), \bigwedge_{2\le i<j \le k} S_{i,j}(x_i, x_j)$.
To compute this query, we will need $k-1$ rounds. In the first round, we will perform in parallel $k-1$ semi-joins to join each unary relation $S'_{1,i}(x_i)$ with one binary relation that contains the variable $x_i$. From Corollary~\ref{lem:semi-join}, this requires a load of $\tilde{O}(M/p')
 = \tilde{O}(M/p^{1-1/k})$, which is less than $\tilde{O}(M/p^{2/k})$ for $k \geq 3$. The resulting
 query is $q'' = K_{k-1}$. Using the inductive hypothesis, we can compute $q''$ with $p'$ servers in $k-2$ rounds with load $\tilde{O}(M/p'^{2/(k-1)})$. To conclude the proof, we calculate that $p'^{2/(k-1)} =  p^{(1-1/k) \cdot 2/(k-1)} = p^{2/k}$. 
\end{proof}

\subsubsection{Other Conjunctive Queries}

Let $q$ be a query that contains an atom $R$, such that $\vars{R} = \vars{q}$; in other words,
suppose that every variable in the body of $q$ appears in the atom $R$. We will show that we can 
compute any query $q$ with such a property in two rounds with optimal load $\tilde{O}(M/p)$.
Notice that the generalized semi-join query $W_k$ has this property, and hence we can
compute $W_k$ with load $\tilde{O}(M/p)$ using two rounds (whereas using one round the
load is $\Omega(M/p^{1/k})$).

\begin{lemma}
Let $q$ be a query that contains an atom $R$, such that $\vars{R} = \vars{q}$. 
Then, $q$ can be computed by an \mpc\ algorithm with  $\tilde{O}(M/p)$ load.
\end{lemma}

The algorithm works as follows. In the first round, we compute in parallel the subqueries 
$q_j = R(x_1, \dots, x_k),  S_j(\dots)$, for every relation $S_j$ that is not $R$. Since each
$q_j$ is a semi-join query, we can apply Corollary~\ref{lem:semi-join} and thus compute all
$q_j$'s in one round with load $\tilde{O}(M/p)$. 
In the second round, we compute the query $\bigwedge_j q_j$, which is an intersection
among the $\ell-1$ intermediate relations $q_j$, where each relation has size at most $m$ tuples. The latter computation can be performed in one round with load $O(M/p)$.

\section{Applications to the External Memory Model}
\label{sec:ex-memory}

In the external memory model, we model computation in the setting where the input data does
not fit into main memory, and the dominant cost is reading the data from the disk into the memory
and writing data on the disk.

Formally, we have an external memory (disk) of unbounded size, and an internal memory (main memory) that consists of $W$ {\em words}.\footnote{The size of the main memory is typically denoted by $M$, but we use $W$ to distinguish from the relation size in the previous sections.} The processor can only use data stored in the internal memory to perform computation, and data can be moved between the two memories in blocks of $B$ consecutive words. The {\em I/O complexity} of an algorithm is the number of input/output blocks that are moved during the algorithm, both from the internal memory to the external memory, and vice versa.

The external memory model has been recently used in the context of databases to 
analyze algorithms for large datasets that do not fit in the main memory, 
with the main application being {\em triangle listing}
\cite{ChuC12,  HuTC13, PaghS14, HuQT15a}.
In this setting, the input is an undirected graph, and the 
goal is to list all triangles in the graph. In~\cite{PaghS14} and~\cite{HuQT15a}, 
the authors consider the related problem of {\em triangle enumeration}, where instead 
of listing triangles (and hence writing them to the external memory), for each triangle in the 
output we call an $emit()$ function. The best result comes from~\cite{HuQT15a}, 
where the authors design a deterministic algorithm that enumerates triangles in 
$O(|E|^{3/2}/(\sqrt{W}B))$ I/Os, where $E$ is the number of edges in the graph. The
authors in~\cite{HuQT15a} actually consider a more general class of join problems, the
so-called {\em Loomis-Whitney enumeration}. In~\cite{Silvestri14},
the author presents external memory algorithms for enumerating subgraph patterns
in graphs other than triangles. More recent work~\cite{HY16} extends the study of I/O complexity to the class of acyclic conjunctive queries (which includes the line query and the star join). The authors present worst-case I/O-optimal algorithms for all acyclic queries when the relations have the same size, and also show optimality for some queries (e.g. $L_3, L_4$) when the sizes are not the same.

The problem we consider in the context of external memory algorithms is a generalization of
triangle enumeration. Given a full conjunctive query $q$, we want to {\em enumerate} all
possible tuples in the output, by calling the $emit()$ function for each tuple in the output
of query $q$. We assume that each tuple in the input can be represented by a single word.

\subsection{Simulating an \mpc\ Algorithm}

We will show how a parallel algorithm in the tuple-based \mpc\ model can help us
construct an external memory algorithm. The {\em tuple-based \mpc\ model} is a restriction of
the \mpc\ model, where only tuples from subqueries of $q$ can be communicated, and
moreover the communication can take a very specific form:  each tuple
$t$ during round $k$ is sent to a set of servers $\mathcal{D}(t,k)$, 
where $\mathcal{D}$ depends only on the data statistics that are initially available to the algorithm.  
Such statistical information is the size of the 
relations, or information about the heavy hitters in the data.\footnote{Even if this
information is not available initially to the algorithm, we can easily obtain it by performing
a single pass over the input data, which will cost $O(|I|/B)$ I/Os.} 
All of the algorithms that we have presented so far in the previous sections satisfy the above assumption. 

The idea behind the construction is that the distribution of the data to the servers can be used
to decide which input data will be loaded into memory; hence, the load $L$ will correspond
to the size of the internal memory $W$. Similarities between hash-join algorithms used for 
parallel processing and the variants of hash-join used for out-of-core processing have
been already known, where the common theme is to create partitions and then process 
them one at a time. Here we generalize this idea to the processing of any conjunctive query
in a rigorous way. We should also note that previous work~\cite{PEM} has studied the simulation
of MapReduce algorithms on a parallel external memory model.

Let $\mA$ be a tuple-based \mpc\  algorithm that computes query $q$ over input $I$ using
$r$ rounds with  load $L(I,p)$. We show next how to construct an external
memory algorithm $\mB$ based on the algorithm $\mA$.

\introparagraph{Simulation}
The external memory algorithm $\mB$ simulates the computation of algorithm $\mA$
during each of the $r$ rounds: round $k$, for $k=1, \dots, r$ simulates the total computation 
of the $p$ servers during round $k$ of $\mA$. We pick a parameter $p$ for the number of
servers that we show how to compute later. The algorithm will store tuples of the form $(t,s)$
to denote that tuple $t$ resides in server $s$.

To initialize $\mB$, we first assign the input data to the $p$ servers (we can
do this in any arbitrary way, as long as the data is equally distributed). 
More precisely, we read each tuple $t$ of the input relations and
then produce a tuple $(t,s)$, where $s= 1, \dots, p$ in a round-robin fashion, such that in the
end each server is assigned $|I|/B$ data items. To achieve this, we load each relation in 
chunks of size $B$ in the memory.
After the initialization, the algorithm $\mB$, for each round $k=1, \dots, r$, 
performs the following steps:

\begin{enumerate}
\item All tuples, which will be of the form $(t,s)$, are sorted according to the
 attribute $s$.
\item All tuples are loaded in memory in chunks of size $W$, in the order by which they 
were sorted in the external memory. If we choose $p$ such that $r \cdot L(I,p) \leq W$, we can
fit in the internal memory all the tuples of any server $s$ at round $k$. 
\footnote{The quantity $L(I,p)$ measures
the maximum amount of data received during any round. Since data is not
destroyed, over $r$ rounds a server can receive as much as $r \cdot L(I,p)$ data.
All of this data must fit into the memory of size $W$, since the decisions of each
server depend on all the data received.}
Hence, we first read into
the internal memory the tuples for server 1, then server 2, and so on.
For each server $s$, we replicate in the internal memory the execution of algorithm $\mA$ in server $s$ at round $k$.
\item For each tuple $t$ in server $s$ (including the ones that are newly produced), we
compute the tuples $\setof{(t,s')}{s' \in \mathcal{D}(t,k)}$, and we write them into the external
memory in blocks of size $B$.
\end{enumerate}

In other words, writing to the internal and external memory simulates the 
communication step, where data is exchanged between servers. The algorithm $\mB$ produces
the correct result, since by the choice of $p$ we guarantee that we can load
enough data in the memory to simulate the local computation of $\mA$ at each server.
Observe that we do not need to write the final result back to the external memory, since at the end
of the last round we can just call $emit()$ for each tuple in the output.

Let us now identify the choice for $p$; recall that we must make sure that $r \cdot L(I,p) \leq W $.
Hence, we must choose $p_o$ such that $ p_{o} = \min_{p} \{ L(I,p) \leq W/r \}$.
We next analyze the I/O cost of algorithm $\mB$ for this choice of $p_o$. 

\introparagraph{Analysis}
The initialization I/O cost for the algorithm is $|I|/B$. 
To analyze the cost for a given round $k=1, \dots, r$, we will measure first the size of the data that will be sorted and then
loaded into memory at round $k$. For this, observe that at every round of algorithm $\mB$,
the total amount of data that is communicated is at most $p_o \cdot L(I, p_o)$. Hence, the total
amount of data that will be loaded into memory will be at most 
$k \cdot p_o \cdot L(I, p_o) \leq p_o W $, from our definition of $p_o$.

For the first step that requires sorting the data, we will not use a sorting algorithm, but instead we will
partition the data into $p$ parts, and then concatenate the parts (this is possible only if $p_o$ is smaller than the memory $W$, i.e. it must be $p_o \leq W$). 
We can do this with a cost of $O(p_o W / B)$ I/Os. The second step of
loading the tuples into memory has a cost of $p_o W /B$, since we are loading the data using
chunks of size $B$; we can do this since the data has been sorted according to the destination
server. As for the third step of writing the data into the external memory, observe that the total
number of tuples written will be equal to the number of tuples communicated to the servers
at round $k+1$, which will be at most $p_o L(I,p_o) \leq p_o W/r$. Hence, the I/O cost will be
$p_o W / (r B)$.

Summing the I/O cost of all three steps over $r$ rounds, 
we obtain that the I/O cost of the constructed algorithm $\mB$ will be:
\begin{align*}
 O \left(  \frac{|I|}{B} + \sum_{k=1}^r \left( \frac{ p_o W}{B} +  \frac{p_o W}{r B}  \right) \right) 
 =  O \left( \frac{|I|}{B} + \frac{rp_o W }{B} \right)
\end{align*}

We have thus proved the following theorem:

\begin{theorem} \label{thm:simulation}
Let $\mA$ be a tuple-based \mpc\ algorithm that computes query $q$ over input $I$ using
$r$ rounds with load $L(I,p)$. For internal memory size $W$, let $ p_{o} = \min_{p} \{ L(I,p) \leq W/r \}$.
If $W \geq p_o$, then there exists an external memory algorithm $\mB$ that
computes $q$ over the same input $I$ with  I/O cost:
$$ O \left( \frac{|I|}{B} + \frac{r p_o W }{B} \right)$$
\end{theorem}

We can simplify the above I/O cost further in the context of 
computing conjunctive queries. In all of our algorithms we used a constant number of rounds
$r$, and the load is typically $L(I,p) \geq |I|/p$. Then, we can rewrite the I/O cost as
$O \left( p_o W / B \right)$.

We can apply~\autoref{thm:simulation} to any of the optimal multi-round algorithms we
presented in the previous sections, and obtain state-of-the-art external memory algorithms
for several classes of conjunctive queries. We show next two applications of our results.


\begin{example}
We presented a 2-round algorithm that computes triangles for
any input data with load (in tuples) $L = \tilde{O}(m/p^{3/2})$, 
in the case where all relations have size $m$.
By applying~\autoref{thm:simulation}, we obtain an external memory algorithm that
computes triangles with $\tilde{O}(m^{3/2}/(BW^{1/2}))$ I/O cost for any $W \geq m^{2/5}$. Notice that this cost
matches the I/O cost for triangle computation from~\cite{PaghS14} up to polylogarithmic
factors.
\end{example}

\begin{example}
We presented a constant-round algorithm that computes any line query $L_k$ with load (in tuples) $L = \tilde{O}(m/p^{1/\lceil (k+1)/2\rceil})$, 
in the case where all relations have size $m$.
By applying~\autoref{thm:simulation}, we obtain an external memory algorithm that
computes the line query $L_k$ with $\tilde{O} \left( \left(\frac{m}{W}\right)^{\lceil (k+1)/2\rceil} \frac{W}{B} \right)$ I/O cost. The cost obtained matches the I/O cost of the algorithm presented in~\cite{HY16} for line queries for relations with equal sizes. 
\end{example}

\section{Conclusion}
\label{sec:conclusion}

In this work, we present the first worst-case analysis for parallel algorithms
that compute conjunctive queries, using the \mpc\ model as the theoretical framework
for the analysis. We also show an interesting connection with the external memory computation model, which 
allows us to translate many of the techniques from the parallel setting to obtain
algorithms for conjunctive queries with (almost) optimal I/O cost. 

The central remaining open question is to design worst-case optimal algorithms for multiple rounds 
for any conjunctive query. We also plan to investigate further the connection between the parallel
setting and external memory setting. It is an interesting question whether our
techniques can lead to optimal external memory algorithms for any conjunctive query, and also 
whether we can achieve a reverse simulation of external memory algorithms in the \mpc\ model.

\noindent \textbf{Acknowledgements.} We would like to thank Ke Yi for pointing out an error in an earlier version of this paper regarding the computation of the edge quasi-packing of the query $L_k$.

\bibliographystyle{plain}
\bibliography{bib}

\begin{thebibliography}{10}

\bibitem{DBLP:journals/corr/abs-1206-4377}
F.~N. Afrati, A.~D. Sarma, S.~Salihoglu, and J.~D. Ullman.
\newblock Upper and lower bounds on the cost of a map-reduce computation.
\newblock {\em CoRR}, abs/1206.4377, 2012.

\bibitem{DBLP:conf/edbt/AfratiU10}
Foto~N. Afrati and Jeffrey~D. Ullman.
\newblock Optimizing joins in a map-reduce environment.
\newblock In {\em EDBT}, pages 99--110, 2010.

\bibitem{DBLP:conf/focs/AtseriasGM08}
A.~Atserias, M.~Grohe, and D.~Marx.
\newblock Size bounds and query plans for relational joins.
\newblock In {\em FOCS}, pages 739--748, 2008.

\bibitem{BKS13}
Paul Beame, Paraschos Koutris, and Dan Suciu.
\newblock Communication steps for parallel query processing.
\newblock In {\em PODS}, pages 273--284, 2013.

\bibitem{BKS14}
Paul Beame, Paraschos Koutris, and Dan Suciu.
\newblock Skew in parallel query processing.
\newblock In {\em PODS}, pages 212--223, 2014.

\bibitem{ChuC12}
Shumo Chu and James Cheng.
\newblock Triangle listing in massive networks.
\newblock {\em {TKDD}}, 6(4):17, 2012.

\bibitem{DBLP:conf/osdi/DeanG04}
Jeffrey Dean and Sanjay Ghemawat.
\newblock Mapreduce: Simplified data processing on large clusters.
\newblock In {\em OSDI}, pages 137--150, 2004.

\bibitem{DBLP:journals/talg/FeldmanMSSS10}
J.~Feldman, S.~Muthukrishnan, A.~Sidiropoulos, C.~Stein, and Z.~Svitkina.
\newblock On distributing symmetric streaming computations.
\newblock {\em ACM Transactions on Algorithms}, 6(4), 2010.

\bibitem{PEM}
Gero Greiner and Riko Jacob.
\newblock The efficiency of mapreduce in parallel external memory.
\newblock In {\em Proceedings of the 10th Latin American International
  Conference on Theoretical Informatics}, LATIN'12, pages 433--445, Berlin,
  Heidelberg, 2012. Springer-Verlag.

\bibitem{myria}
Daniel Halperin, Victor~Teixeira de~Almeida, Lee~Lee Choo, Shumo Chu, Paraschos
  Koutris, Dominik Moritz, Jennifer Ortiz, Vaspol Ruamviboonsuk, Jingjing Wang,
  Andrew Whitaker, Shengliang Xu, Magdalena Balazinska, Bill Howe, and Dan
  Suciu.
\newblock Demonstration of the {Myria} big data management service.
\newblock In Curtis~E. Dyreson, Feifei Li, and M.~Tamer {\"{O}}zsu, editors,
  {\em International Conference on Management of Data, {SIGMOD} 2014, Snowbird,
  UT, USA, June 22-27, 2014}, pages 881--884. {ACM}, 2014.

\bibitem{HY16}
Xiao Hu and Ke~Yi.
\newblock Towards a worst-case {I}/{O}-optimal algorithm for acyclic joins.
\newblock In {\em PODS}, 2016.

\bibitem{HuQT15a}
Xiaocheng Hu, Miao Qiao, and Yufei Tao.
\newblock Join dependency testing, {Loomis-Whitney} join, and triangle
  enumeration.
\newblock In {\em Proceedings of the 34th {ACM} Symposium on Principles of
  Database Systems, {PODS} 2015, Melbourne, Victoria, Australia, May 31 - June
  4, 2015}, pages 291--301, 2015.

\bibitem{HuTC13}
Xiaocheng Hu, Yufei Tao, and Chin{-}Wan Chung.
\newblock Massive graph triangulation.
\newblock In {\em Proceedings of the {ACM} {SIGMOD} International Conference on
  Management of Data, {SIGMOD} 2013, New York, NY, USA, June 22-27, 2013},
  pages 325--336, 2013.

\bibitem{DBLP:conf/soda/KarloffSV10}
H.~J. Karloff, S.~Suri, and S.~Vassilvitskii.
\newblock A model of computation for mapreduce.
\newblock In {\em SODA}, pages 938--948, 2010.

\bibitem{Klauck15}
Hartmut Klauck, Danupon Nanongkai, Gopal Pandurangan, and Peter Robinson.
\newblock Distributed computation of large-scale graph problems.
\newblock In {\em Proceedings of the Twenty-Sixth Annual ACM-SIAM Symposium on
  Discrete Algorithms}, SODA '15, pages 391--410. SIAM, 2015.

\bibitem{DBLP:conf/pods/KoutrisS11}
P.~Koutris and D.~Suciu.
\newblock Parallel evaluation of conjunctive queries.
\newblock In {\em PODS}, pages 223--234, 2011.

\bibitem{dremel}
S.~Melnik, A.~Gubarev, J.~J. Long, G.~Romer, S.~Shivakumar, M.~Tolton, and
  T.~Vassilakis.
\newblock Dremel: Interactive analysis of web-scale datasets.
\newblock {\em PVLDB}, 3(1):330--339, 2010.

\bibitem{NPRR12}
H.~Q. Ngo, E.~Porat, C.~R{\'e}, and A.~Rudra.
\newblock Worst-case optimal join algorithms: [extended abstract].
\newblock In {\em PODS}, pages 37--48, 2012.

\bibitem{PaghS14}
Rasmus Pagh and Francesco Silvestri.
\newblock The input/output complexity of triangle enumeration.
\newblock In {\em PODS}, pages 224--233, 2014.

\bibitem{Silvestri14}
Francesco Silvestri.
\newblock Subgraph enumeration in massive graphs.
\newblock {\em CoRR}, abs/1402.3444, 2014.

\bibitem{WZ13}
DavidP. Woodruff and Qin Zhang.
\newblock When distributed computation is communication expensive.
\newblock In Yehuda Afek, editor, {\em Distributed Computing}, volume 8205 of
  {\em Lecture Notes in Computer Science}, pages 16--30. Springer Berlin
  Heidelberg, 2013.

\bibitem{spark:2012}
M.~Zaharia, M.~Chowdhury, T.~Das, A.~Dave, J.~Ma, M.~McCauley, M.~J. Franklin,
  S.~Shenker, and I.~Stoica.
\newblock Resilient distributed datasets: a fault-tolerant abstraction for
  in-memory cluster computing.
\newblock In {\em NSDI}, 2012.

\end{thebibliography}

\appendix

\newpage
\section{Appendix A: Proof of One-Round Optimal Lower Bound}
\label{sec:appendix}

\begin{proof}[Proof of~\autoref{th:lower:skew}]
We construct a probability space for instances $I$ defined by $\bM$ and the
choice of $\bx$ as follows. 
First, let $\alpha$ be a constant value.
We pick the domain $n$ such that $n = (\max_j \{m_j\})^2$.
For each relation $S_j$, we construct an instance by choosing a uniformly
random instance over all matching relations on the variables in
$\bx_j = \vars{S_j} \setminus \bx$, while fixing the variables $\bx$ to the constant $\alpha$.
In this construction, for every $\ba_j \in [n]^{\bx_j}$, the probability that $S_j$ contains
$\ba_j$ is $P(\ba_j \in S_j) = m_j / n^{a_j-d_j}$. If $a_j = d_j$, and thus $\bx$ includes all
variables from $S_j$, then the instance contains the single tuple $(\alpha, \dots, \alpha)$,
and we also fill the instance with arbitrary tuples of fresh values.

We can show that the expected number of answers will be
$$ \E[|q(I)|] =   n^{k-d} \prod_{j: S_j \in q_{\bx} } \frac{m_j}{ n^{a_j-d_j}}$$

Let us fix some server and let $\msg(I)$ be the message the server receives on input $I$.
Let $K_{\msg_j}(S_j)$ denote the set of tuples from relation $S_j$
{\em known by the server}. Let $w_j({\ba_j}) = P(\ba_j \in K_{\msg_j(S_j)}(S_j))$, 
where the probability is over the random choices of $S_j$.  This is upper bounded
by $P(\ba_j \in S_j)$: $w_j(\ba_j) \leq m_j / n^{a_j-d_j}$.

\begin{lemma}
\label{lem:L-bound}
Let $S_j$ a relation with $a_j >d_j$.
Suppose that the size of $S_j$ is $m_j \leq n/2$ (or $m_j = n$), and that the 
message $\msg_j(S_j)$ has at most $L$ bits. Then,
$\E[|K_{\msg_j}(S_j)|] \leq \frac{4 L}{(a_j-d_j) \log(n)}$.
\end{lemma}

\begin{proof}
We can express the entropy $H(S_j)$ as follows:
\begin{align} \label{eq:entropy:one}
 H(S_j) &=H(\msg_j(S_j))+ \sum_{\msg_j}P(\msg_j(S_j)=\msg_j) \cdot H(S_j \mid \msg_j(S_j)=\msg_j) \nonumber\\
 &\le L +  \sum_{\msg_j}P(\msg_j(S_j)=\msg_j) \cdot H(S_j \mid \msg_j(S_j)=\msg_j) 
\end{align}

Denoting by $\mathcal{M}_j$ the number of bits
necessary to represent  $S_j$, we have:
\begin{align*}
 H(S_j \mid \msg_j(S_j)=\msg_j)  
 & \leq \left(1-\frac{|K_{\msg_j}(S_j)|}{2 m_j} \right) \mM_j\\
 & \leq H(S_j) - \frac{|K_{\msg_j}(S_j)|}{2 m_j} m_j \frac{a_j-d_j}{2} \log (n)  \\ 
 & = H(S_j) - (1/4) \cdot |K_{\msg_j}(S_j)| (a_j-d_j)\log (n) 
\end{align*}
where the last inequality comes from~\cite{BKS14}.
Plugging this in~\eqref{eq:entropy:one}, and solving for $\E[|K_{\msg_j}(S_j)|] $:
\begin{align*}
\E[|K_{\msg_j}(S_j)|]  \leq \frac{4 L}{(a_j-d_j)\log(n)}
\end{align*}
This concludes our proof.
\end{proof}

Define the {\em extended query} ${q_{\bx}}'$ to consist of $q_\bx$, where we add a
new atom $S_i'(x_i)$ for every variable $x_i \in
vars(q_\bx)$.  Define $u'_i = 1 - \sum_{j: i \in S_j} u_j$.
In other words, $u'_i$ is defined to be the slack at the variable
$x_i$ of the packing $\bu$.  The new edge packing
$(\mathbf{u},\mathbf{u}')$ for the extended query $q_\bx'$ has no more
slack, hence it is both a tight fractional edge packing and a tight
fractional edge cover for $q_{\bx}$.  By adding all equalities of the
tight packing we obtain:
$$\sum_{j=1}^{\ell} (a_j-d_j) u_j + \sum_{i=1}^{k-d} u'_i = k-d$$

We next compute how many output tuples from $q(I)$ will be known
in expectation by the server. We have:
\begin{align*}
\E[|K_{\msg}(q(I))|] 
& = \sum_{\ba \in [n]^d} \prod_{j: S_j \in q_\bx} w_j({\ba_j}) \\
& =  \sum_{\ba \in [n]^d} \prod_{j: S_j \in q_\bx} w_j({\ba_j})
         \prod_{i=1}^{k-d} w_i'({\ba_i})\\      
 & \leq   \prod_{i=1}^{k-d} n^{u_i'} \cdot 
  \prod_{j: S_j \in q_\bx} \left( \sum_{\ba_j \in [n]^{d_j}} w_j({\ba_j})^{1/u_j} \right)^{u_j} 
\end{align*}
By writing $w_j(\ba_j)^{1/u_j} = w_j(\ba_j)^{1/u_j-1} w_j({\ba_j})$,  we can bound the sum in the above quantity as follows:
\begin{align*}
 \sum_{\ba_j \in [n]^{d_j}} w_j({\ba_j})^{1/u_j} 
 & \leq \left( \frac{m_j}{ n^{a_j-d_j}} \right)^{1/u_j-1} \sum_{\ba_j \in [n]^{d_j}} w_j({\ba_j}) 
\quad  = (m_j   n^{d_j-a_j})^{1/u_j-1}  L_j
\end{align*}
where $L_j = \sum_{\ba_j \in [n]^{d_j}} w_j({\ba_j})$. 
We can now write:
\begin{align}
\E[|K_{\msg}(q(I))|] 
& \leq  n^{\sum_{i=1}^{k-d} u_i'}  \prod_{j:S_j \in q_\bx} \left( L_{j}   m_j^{1/u_j-1}  n^{(d_j-a_j)(1/u_j-1)} \right)^{u_j}  \nonumber\\
& =  \prod_{j:S_j \in q_\bx} L_{j}^{u_j} \cdot \prod_{j:S_j \in q_\bx} m_j^{-u_j} \cdot \E [|q(I)|] \\
& \leq \prod_{j:S_j \in q_\bx} \left( \frac{4L}{m_j (a_j -d_j) \log(n)} \right)^{u_j} \cdot \E [|q(I)|] 
\label{eq:lastline}
\end{align}

Using the fact that $M_j = a_j m_j \log(n)$, and that the $p$ servers produce all answers,
we obtain that, in order to obtain all answers, we must have
$$
 L \geq  \min_j \dfrac{a_j-d_j}{4 a_j}  \cdot \left( \frac{\prod_{j:S_j \in q_\bx}  M_j^{u_j}}{p} \right)^{1/\sum_j u_j}  
$$
which concludes the proof. 
\end{proof}

\end{document}